\documentclass[conference]{IEEEtran}

\usepackage{cite}
\usepackage[cmex10]{amsmath}
\usepackage{array}
\usepackage{bm}
\usepackage{mdwmath}
\usepackage{mdwtab}
\usepackage[tight,footnotesize]{subfigure}
\usepackage{tikz}
\usepackage{mathtools}
\usetikzlibrary{arrows}
\usetikzlibrary{matrix,decorations.pathreplacing}

\usepackage[T1]{fontenc}
\usepackage{amssymb,amsfonts,amsthm}
\usepackage{graphicx,epsfig,psfrag,color}
\usepackage{times}
\usepackage{dsfont}
\usepackage{algorithm}
\usepackage[noend]{algorithmic}
\usepackage{caption}

\usepackage{array}
\usepackage{mdwmath}
\usepackage{mdwtab}
\usepackage{tikz}
\usepackage{graphicx,epsfig,psfrag,color}
\usepackage{multirow}

\usetikzlibrary{decorations.pathreplacing}

\tikzstyle{vertex}=[circle,fill=black!15,minimum size=20pt,inner sep=0pt,font=\footnotesize]
\tikzstyle{smallvertex}=[vertex,minimum size=10pt]
\tikzstyle{operator}=[vertex,fill=black!1]
\tikzstyle{smalloperator}=[circle,inner sep=0pt,minimum size=10pt,font=\footnotesize]
\tikzstyle{source} = [vertex, fill=red!34]
\tikzstyle{supersource} = [vertex, fill=blue!34,minimum size=15pt]
\tikzstyle{hiddensource} = [vertex, fill=red!8,minimum size=10pt]
\tikzstyle{smallsource} = [vertex, fill=red!34,minimum size=10pt]
\tikzstyle{receiver} = [vertex, fill=green!34]
\tikzstyle{smallreceiver} = [vertex, fill=green!34,minimum size=10pt]
\tikzstyle{edge} = [draw,thick,->]
\tikzstyle{undirect_edge} = [draw, thick]
\tikzstyle{dedge} = [edge,dotted]
\tikzstyle{redge} = [edge,color=red]
\tikzstyle{bedge} = [edge,color=blue]
\tikzstyle{gedge} = [edge,color=green]
\tikzstyle{medge} = [edge,color=magenta]
\tikzstyle{oedge} = [edge,color=orange]
\tikzstyle{bredge} = [edge,color=brown,line width=2pt]
\tikzstyle{weight} = [font=\footnotesize]
\tikzstyle{selected edge} = [draw,line width=5pt,-,red!50]

\newcommand{\bi}{\begin{itemize}}
\newcommand{\ei}{\end{itemize}}
\newcommand{\bal}{\begin{align}}
\newcommand{\eal}{\end{align}}
\newcommand{\EE}{\mathbb{E}}
\newcommand{\PP}{\mathbb{P}}

\newcommand{\bA}{\mathbf{A}}
\newcommand{\bQ}{\mathbf{Q}}
\newcommand{\bN}{\mathbf{N}}
\newcommand{\bD}{\mathbf{D}}

\newcommand{\bI}{\mathbf{I}}
\newcommand{\bP}{\mathbf{P}}

\newtheorem{theorem}{Theorem}
\newtheorem{lemma}{Lemma}
\newtheorem{corollary}{Corollary}
\theoremstyle{definition}
\newtheorem{definition}{Definition}
\theoremstyle{remark}
\newtheorem{remark}{Remark}

\begin{document}

\title{Maximum Entropy Functions: Approximate G\'acs-K\"orner for  Distributed Compression}

\author{\IEEEauthorblockN{Salman Salamatian}
\IEEEauthorblockA{MIT}
\and
\IEEEauthorblockN{Asaf Cohen}
\IEEEauthorblockA{Ben-Gurion University of the Negev}
\and
\IEEEauthorblockN{Muriel M\'edard}
\IEEEauthorblockA{MIT}}

\maketitle

\begin{abstract}
Consider two correlated sources $X$ and $Y$ generated from a joint distribution $p_{X,Y}$. Their G\'acs-K\"orner Common Information, a measure of common information that exploits the combinatorial structure of the distribution $p_{X,Y}$, leads to a source decomposition that exhibits the latent common parts in $X$ and $Y$. Using this source decomposition we construct an efficient distributed compression scheme, which can be efficiently used in the network setting as well. Then, we relax the combinatorial conditions on the source distribution, which results in an efficient scheme with a helper node, which can be thought of as a front-end cache. This relaxation leads to an inherent trade-off between the rate of the helper and the rate reduction at the sources, which we capture by a notion of optimal decomposition. We formulate this as an approximate G\'acs-K\"orner optimization. We then discuss properties of this optimization, and provide connections with the maximal correlation coefficient, as well as an efficient algorithm, both through the application of spectral graph theory to the induced bipartite graph of $p_{X,Y}$.

\end{abstract}


\IEEEpeerreviewmaketitle

\section{Introduction}

Consider two distributed sources $X$ and $Y$ generated from a joint distribution $p_{X,Y}$. The fundamental compression limits of these distributed sources is given by the Slepian-Wolf rate region \cite{Slepian:1973wj}. While the fundamental limits are well known, and some low-complexity codes exist, implementations of Slepian-Wolf codes in practical scenarios are seldom seen. Indeed, most present-day systems use naive solutions to avoid redundancy such as deduplication. This is partly explained by more realistic problem instances, for example, one in which the sources are distributed in a network. In such a setting, distributed compression cannot be separated from the network operations, and should be done hand in hand to achieve the optimal rates \cite{Ramamoorthy:2006vr}. This brings complexity issues, as well as the implementation burden, which makes the optimal Slepian-Wolf codes not well suited for the task.

In this paper, we propose to look at a sub-optimal, but very efficient distributed coding technique that uses the combinatorial structure of the correlation for compression. Namely, we try to find a \emph{common part} in the sources $X$ and $Y$, and send this common part only once to the terminal. We then express $X$ and $Y$ given that common part, which result in an overall reduction of the necessary rates. It turns out that the right common part to analyze is the G\'acs-K\"orner common information, introduced in \cite{Gacs:1973vg}. This measure of common-information has a strong relationship with the combinatorial structure of the joint distribution $p_{X,Y}$, and in particular to its bipartite representation. However, the G\'acs-K\"orner common information does not capture, in general, the complete correlation between the sources, an restrict the joint distributions for which such a scheme is possible.
Instead, we relax this combinatorial condition, which leads to a scheme in which a helper is required. The role of the helper is to complement the distributed encoders if the common part is not exactly the same at both sources. By increasing the rate of the helper to the terminal, we can achieve better and better rate reduction.

This type of approach is particularly well suited to the network setting. As we are decomposing our sources into common and non-common parts, and transmit these latent parts through the network, we can use efficient network coding techniques without affecting the source code --- the decomposed sources are independent and the rate reduction comes from solely the decomposition. Moreover, very efficient single-source lossless coding can be used. 

The helper node can be thought of as a \emph{front-end cache}, which is well connected to the sources and is being used to significantly reduce the rates required in the last-mile to the terminals. Thus, when laying out a network, our result is useful in determining how to position the helper so it acts as a valuable interface between the sources and the terminals, \emph{which function of the data to store on it}, and what are the required rates.

\noindent \textbf{Main Contributions:} Our contributions are three-fold. First, we introduce a distributed coding technique based on the structure of the correlation between the two sources $X$ and $Y$, that we describe in terms of functions $\phi_X$ and $\phi_Y$. We start with the case in which a helper is not necessary, \emph{i.e.}, the functions $\phi_X$ and $\phi_Y$ always agree. The common part is then given by the G\'acs-K\"orner Common information. Later, we introduce a helper node, which allows us to relax this agreement condition. For each case, we describe a simple and efficient coding method, and compute its achievable region.\\
These regions reveal an inherent trade-off between the rate of the helper and the amount of rate reduction we can obtain at the sources. This trade-off is captured by the notion of \emph{optimal decomposition}. We describe optimal decomposition as an optimization problem over functions, give it an approximate G\'acs-K\"orner formulation and discuss some of its properties in the context of coding and optimization.
Finally, as the problem of finding an optimal decomposition is in general hard, we provide an approximate algorithm based on spectral graph theory, and draw some connections with the maximal correlation coefficient.

\section{Background}

\subsection{Notation}

Throughout, we define finite random variables with capital letters, \emph{e.g.} X, taking values in the finite set $[n_X]:= \{1,\ldots,n_X\}$, where $n_X$ is the size of the alphabet.
We will be interested in correlated random variables $X$ and $Y$, taking value from a joint distribution $p_{X,Y}$. We let $\bP$ be the $n_X \times n_Y$ joint distribution matrix, that is $\{\bP\}_{i,j} = p_{X,Y}(i,j)$. We denote by $\mathbf{1}$ the vector of all ones whose size will be clear from the context. We also denote by $\mathbb{I}\{.\}$ the indicator function.

We will represent networks in the following way: let $\mathcal{G} = (\mathcal{V},\mathcal{E},\mathcal{C})$ be an acyclic directed graph with edge set $\mathcal{E}$, vertices $\mathcal{V}$ and an integer valued function on each edge, $\mathcal{C} : E \rightarrow \mathbb{N}$. The value $C(e)$ for $e \in \mathcal{E}$ represents the capacity of the communication link $e$ in bits per unit of time. In addition, let $s_X,s_Y \in \mathcal{V}$ be two sender nodes and let $t$ be a terminal node. We will denote values of min-cuts in the network by $\rho(\cdot;\cdot)$. For example, $\rho(s_1;t_1)$ represents the value of the min-cut from $s_1$ to $t_1$, and $\rho(s_1, s_1;t_2)$ the value of the min-cut from both $s_1$ and $s_2$ to $t_2$. 

\subsection{Disjoint Components and Common Information}

A useful point of view for our problem will be to look at \emph{the bipartite representation} of the joint distribution probability $p_{X,Y}$, which is defined below:

\begin{definition}
Let $p_{X,Y}$ be the joint probability distribution of $X,Y$. We denote by the \emph{bipartite} representation of $p_{X,Y}$ the bipartite graph with $n_X + n_Y$ nodes indexed by $v_{a,i}$ where $a \in \{X,Y\}$ and $i \in [n_a]$. We let an edge connect node $v_{X,i}$ to $v_{Y,j}$ only if $p_{X,Y}(i,j) > 0$, and let $p_{X,Y}(i,j)$ be the weight of that edge.
\end{definition}

This bipartite representation allows for a simple characterization of the G\'acs-K\"orner Common Information. Specifically, we have the following definition of common information:

\begin{definition}[\cite{Witsenhausen:1975vk}]
A set of nodes $\mathcal{C}$ such that there are no outgoing edges from $\mathcal{C}$ is called a connected component of the bipartite graph. We associate with each connected component $\mathcal{C}$ a weight $p(\mathcal{C}) = \sum_{n_x,n_y \in C} P(X=x, Y=y)$. We call \emph{the common information decomposition} of $p_{X,Y}$, the decomposition of the bipartite graph into a maximal number of connected components $\mathcal{C}_1,\ldots,\mathcal{C}_{k}$. Moreover, we denote by $K_{X,Y}$ the common information random variable representing the index of the connected component, with the natural distribution $(p(\mathcal{C}_1),\ldots,p(\mathcal{C}_k))$. The entropy of this random variable is $H(K_{X,Y}) = H(\mathbf{\mathcal{C}}) = \sum_{i=1}^k p(\mathcal{C}_i) \log \left( \frac{1}{p(\mathcal{C}_i)}\right)$.
\label{def:common_inf}\end{definition}

Definition \ref{def:common_inf} is equivalent to the usual definition of common information in the 2 random variables setting given in \cite{Gacs:1973vg}:
\begin{align}
K_{X,Y} = \underset{H(U|X) = H(U|Y) = 0}{\text{argmax}} H(U)
\end{align}
It is well known that the above optimization problem achieves its maximum point by letting $U$ be the index of the connected component, we can rewrite the optimization above by restrciting the space to these functions.

We now present an alternative formulation of the above optimization. While being equivalent in the most basic setting, this formulation will allow us to extend the notion of common information to cases where the $K_{X,Y}$ is zero or two small to be practical, and be the basis for future arguments in this paper. Precisely, consider functions $\phi_X$ and $\phi_Y$ mapping $[n_X]$ and $[n_Y]$ respectively, into a finite set. It is then easy to see that the common information then becomes the solution of the following optimization problem:
\begin{align}
\text{max}_{\phi_X,\phi_Y} \, & H(\phi_X(X)), \, \text{s.t.} \, \PP (\phi_X(X) \neq \phi_Y(Y)) = 0. \label{eq:gk_function_formulation}
\end{align}

\noindent where the optimization is over finite range functions $\phi_X, \phi_Y$.

\noindent \textbf{Example: } Let $P_{X_1,X_2}$ be defined as in joint probability table below with its corresponding 2-partite graph representation. Consider the function $f(i)$ that takes value $f(i) = 1$ if $i \in \{ 1,2 \}$, and $f(i) = -1$ if $i \in \{ 3,4\}$. Then the maximizing functions in \eqref{eq:gk_function_formulation} are $\phi_X = \phi_Y = f$.

\begin{minipage}{.45\textwidth}
\begin{minipage}[c]{.45\textwidth}
\centering
\begin{tikzpicture}[scale=1.2, auto,swap]
		\foreach \pos/ \name/ \num in {{(0,0)/x_4/4}, {(0,.5)/x_3/3}, {(0,1)/x_2/2}, {(0,1.5)/x_1/1}, {(2,0)/y_4/4},{(2,.5)/y_3/3},{(2,1)/y_2/2},{(2,1.5)/y_1/1}}
        \node (\name) at \pos {$\num$};
        \node (cal_Y) at (2.3,1.7) {$\mathcal{Y}$};
        \node (cal_X) at (-0.3,1.7) {$\mathcal{X}$};
        \foreach  \source/ \dest /\edgename  in {x_1/y_1/1, x_2/y_2/2, x_1/y_2/3, x_2/y_1/4, x_3/y_3/5, x_4/y_4/6, x_3/y_4/7, x_4/y_3/8}
        \path[undirect_edge] (\source) -- (\dest);
        \end{tikzpicture}
\end{minipage} 
\begin{minipage}[c]{.45\textwidth}
\centering
\begin{align*}
\bP = \frac{1}{8} \left[ \begin{array}{cccc}
1 & 1 & 0 & 0 \\
1 & 1 & 0 & 0 \\
0 & 0 & 1 & 1 \\
0 & 0 & 1 & 1
\end{array}\right]
\end{align*}
\end{minipage}
\end{minipage}
\subsection{Maximal Correlation Coefficient}

We now define an alternative measure of information between random variables $X$ and $Y$ namely, the maximal correlation, first defined in \cite{renyi1959}. For a given joint distribution $p_{X,Y}$, we let $\rho_m(X;Y) \vcentcolon = \sup_{\phi_X,\phi_Y} \EE [\phi_X(X) \phi_Y(Y)]$, where the maximization is taken over all mean zero real-valued functions $f$ and $g$ on $[n_x]$ and $[n_y]$, such that $\EE [\phi_X(X)^2] = \EE [\phi_Y(Y)^2] = 1$. For random variables taking values in a finite set, it has been shown that the maximal correlation coefficient is related to the spectrum of the matrix $\bP$ through the following identity proved in \cite{Witsenhausen:1975vk}:

\begin{theorem}\label{thm:max_corr_def}
Denote by $\bD_X$ the diagonal matrix with diagonal elements $(\bD_X)_{i,i} = p_X(i)$, and define $\bD_Y$ similarly. We let $\bQ \in \mathbb{R}^{n_x \times n_y}$ be the matrix defined as $\bQ \vcentcolon = \bD^{-1/2}_X \bP \bD^{-1/2}_Y$.
Its singular value decomposition is given by $\bQ = \mathbf{U} \mathbf{\Sigma} \mathbf{V}^T$, and we denote by $1 \leq \lambda_1 \leq \ldots \leq \lambda_{d}$ the ordered singular values (we refer the readers to \cite{Witsenhausen:1975vk} or \cite{flavio} to see why the highest singular value always takes value 1). In particular, $\lambda_1 = \rho^2_m(X;Y)$ is the squared maximal correlation coefficient of $p_{X,Y}$.
\end{theorem}
%

\section{An Efficient Distributed Coding Scheme}
\subsection{Using G-K common information}

Consider the classical Slepian-Wolf setting, despicted in Fig~\ref{fig:classic_sw}. In that setting, two separate encoders $E_X$ and $E_Y$ encode at rates $R_X$ and $R_Y$, respectively. The goal is to recover, in a lossless way, the pair of sources $(X,Y)$ at a destination node $D_{X,Y}$. We will be interested in efficient encoding and decoding for this canonical problem. In particular, when there is non-zero common information, \emph{i.e.} $H(K) > 0$, we can use a simple \emph{source decomposition} of the sources $X$ and $Y$.
Namely, let $X \to (X',K)$, where $K$ is the index of the connected component, and $X'$ is the position of $X$ inside that component. Note that this is a bijective function of $X$. Note that at the source $Y$ we can also do a similar decomposition $Y \to (Y',K)$.
Noticing that $K$ is common in both decompositions yield a very simple distributed coding scheme, in which $X'$ and $Y'$ are sent as if independent, but $K$ is sent only once to the receiver. All of these operations can be done efficiently, and the resulting coding is a zero-error distributed code.

\begin{theorem}[Coding with G-K Common Information \cite{technical}]
Let $K$ be the G-K Common information between $X$ and $Y$. Then, there exist an efficient zero error encoding and decoding of $X$ and $Y$ that operates at rates:
\begin{align}
R_X  \geq H(X|K) ,\quad &R_Y \geq H(Y|K) \nonumber, \\
R_X + R_Y  \geq H(X|K) + &H(Y|K) + H(K).
\end{align}
\end{theorem}

Note that the rates in this result depend on the value of the common information. Indeed, the higher $H(K)$, the greater is the reduction in rates. However, as we mentioned before,the common-information imposes a strong combinatorial condition on the joint distribution matrix. The next section aims at relaxing this strong condition by considering a helper node.


\begin{remark}
The proposed coding scheme is different than the coding that exist in zero-error distributed coding, see \emph{e.g.} \cite{koulgi2003zero,witsenhausen1976zero}. In general, the zero-error distributed coding sum-rate is better than the sum rate in the coding scheme above. However, this former type of methods do not generalize very well to any point on the achievable region, as there is generally very little flexibility in the achievable rates. This make these codes less suited for network extensions where the choice of the operating rates is given by the network topology.
\end{remark}
\subsection{General distributions with a helper}

We now aim at moving away from the combinatorial condition imposed on the joint distribution $p_{X,Y}$. In particular, say the bipartite graph is made of two \emph{almost} disjoint components, only connected by an edge of small weight $\epsilon$. In that case, the common information evaluates to 0. However, it would seem like being optimistic and disregarding this edge yield two disjoint components, and therefore a non-zero common information. To correct the problematic cases, we can consider an omniscient helping node, which simply sends a bit in case of an error.
This insight is developed in the theorem below, where we consider arbitrary binary functions of $X$ and $Y$, and decompose the sources based on these functions.

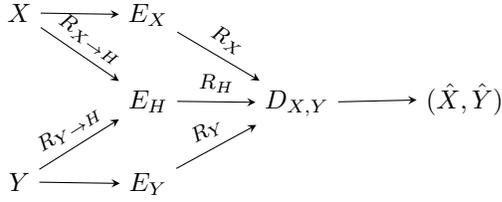
\begin{figure}
\centering
\begin{tikzpicture}
     \matrix (m) [matrix of math nodes, row sep=1.5em, column sep=3em]{
    X & E_X &  &  \\
      & E_{H} &  D_{X,Y} &  (\hat{X},\hat{Y}) \\
    Y & E_Y &  &  \\};   
    \path[-stealth]
    (m-1-1) edge (m-1-2) edge node[anchor = south,sloped] {\footnotesize $R_{X \to H}$} (m-2-2)
    (m-1-2) edge node[anchor = south, sloped] {\footnotesize $R_X$} (m-2-3)
    (m-2-2) edge node[auto] {\footnotesize $R_H$} (m-2-3)
    (m-3-1) edge (m-3-2)  edge node[anchor = south,sloped] {\footnotesize $R_{Y \to H}$} (m-2-2)
    (m-3-2) edge node[anchor = south, sloped] {\footnotesize $R_Y$} (m-2-3)
    (m-2-3) edge (m-2-4);
\end{tikzpicture}
\caption{Slepian-Wolf with Helper. In the classical setting, $R_H = R_{X \to H} = R_{Y \to H} = 0$} \label{fig:classic_sw}
\end{figure}

\begin{theorem}[Binary functions]
Consider $\phi_X$ and $\phi_Y$ taking values in $\{0,1\}$ be functions of $X$ and $Y$, respectively. Let $Pr(\phi_X \neq \phi_Y) = \epsilon$. Then, there exist an efficient zero-error encoding and decoding of $X$ and $Y$ with a helper of rate $R_h = h(\epsilon)$, that operates at rates, for some $0 \leq \alpha \leq 1$:
\begin{align}
R_X  \geq H(X|\phi_X(X)), \quad &  R_Y \geq H(Y|\phi_Y(Y)) \nonumber, \\
R_X + R_Y \geq H(X|\phi_X(X)) + & H(Y|\phi_Y(Y)) +  \nonumber \\
 \alpha H(\phi_X(X)) & + (1-\alpha) H(\phi_Y(Y)) .\label{eq:rates_non_symm}
\end{align}
\end{theorem}

\begin{proof}
We prove the corner points $(H(X|\phi_X(X)),H(Y))$. The rest of the region follows by symmetry and time-sharing.
Consider $n$ iid samples $(x^n,y^n)$, and the corresponding $(\phi_X(x^n),\phi_Y(y^n))$. Denote by $e^n$ the sequence such that $e_i = 1$ if $\phi_X(x_i) \neq \phi_Y(y_i)$ and 0 otherwise.
Then, we can encode $x^n$ given $\phi_Y(x^n)$ using approximately $nH(X|\phi_X(X))$ bits. In contrast, we encode $y^n$ fully using $nH(Y)$ bits, and let the helper encode $e^n$ using $nH(\epsilon)$ bits.
At the receiver, we can decode $x^n$, and $e^n$. From $x^n$ we can obtain $\phi_X(x^n)$ which along with $e^n$ determines uniquely $\phi_Y(y^n)$. Therefore we can obtain $y^n$ as well.
\end{proof}

It is interesting to note that in comparison with the usual Slepian-Wolf region, the rates in \eqref{eq:rates_non_symm} do not describe a \emph{symmetric} achievable region, that is the dominant face is not necessarily a 45 degrees line.
The above theorem can be extended to the setting in which the functions $\phi_X$ and $\phi_Y$ can take value in larger sets. The rate of the helper then depends on the position on the dominant face, as we need to encode which of the values $\phi_X$  and $\phi_Y$ to take if $\phi_X \neq \phi_Y$.


\begin{theorem}[Arbitrary functions]\label{thm:arbitrary_func}
Let $\phi_X : \mathcal{X} \to \{ 1,\ldots, n_x\}$ and $\phi_Y: \mathcal{Y} \to \{ 1,\ldots,n_y\}$. Let the rate of the helper be $H(\phi_X(X)|\phi(Y))$. The following corner point is achievable:
\begin{align}
R_X \geq H(X|\phi_X(X)) , \quad & R_Y \geq H(Y), \nonumber \\
R_h \geq H(\phi_X(X)&|\phi_Y(Y)).
\end{align}
\end{theorem}

This theorem can be extended to obtain a complete rate region, by first considering the symmetric corner point and then time-sharing. Note that in this context, the rate of the helper changes depending on the position in the dominant face. The proof is omitted.

\noindent \textbf{Example:} Let us revisit our previous example. This time however, we let an edge of weight $\frac{\delta}{8}$ join the two components. The common information of the resulting $p_{X,Y}$ is zero. However, we can find a simple decomposition of the sources by considering the same $\phi_X$ and $\phi_Y$ as before, and letting the helper encode the errors. More precisely, let $E_i = \mathbf{1}\{ X _i= x_2 \text{ and } Y_i = y_3\}$. Then, the helper can simply encode $E^n$, using approximately $nH(E) = nh(\frac{\delta}{8})$ bits.
\begin{minipage}{.45\textwidth}
\begin{minipage}[c]{.45\textwidth}
\centering
\begin{tikzpicture}[scale=1.2, auto,swap]
		\foreach \pos/ \name/ \num in {{(0,0)/x_4/4}, {(0,.5)/x_3/3}, {(0,1)/x_2/2}, {(0,1.5)/x_1/1}, {(2,0)/y_4/4},{(2,.5)/y_3/3},{(2,1)/y_2/2},{(2,1.5)/y_1/1}}
        \node (\name) at \pos {$\num$};
        \node (cal_Y) at (2.3,1.7) {$\mathcal{Y}$};
        \node (cal_X) at (-0.3,1.7) {$\mathcal{X}$};
        \foreach  \source/ \dest /\edgename  in {x_1/y_1/1, x_2/y_2/2, x_1/y_2/3, x_2/y_1/4, x_3/y_3/5, x_4/y_4/6, x_3/y_4/7, x_4/y_3/8}
        \path[undirect_edge] (\source) -- (\dest);
        \path (x_2) edge[red] (y_3);
        \end{tikzpicture}
        \end{minipage} 
\begin{minipage}[c]{.45\textwidth}
\centering
\begin{align*}
\bP = \frac{1}{8} \left[ \begin{array}{cccc}
1 & 1 & 0 & 0 \\
1 & 1 -\delta & \delta & 0 \\
0 & 0 & 1 & 1 \\
0 & 0 & 1 & 1
\end{array}\right]
\end{align*}
\end{minipage}
\end{minipage}

It is interesting to note that in the previous example, the helper node does not need to know the value of $X$ and $Y$ if they are not part of the cut. Precisely, the source $X$ only need to send whether the value is $X = 2$. Similarly, at $Y$, we only need to describe whether $Y = 3$. Indeed, this information is enough for the helper to compute whether there was an error. In fact, this observation is not specific to this example, and the helper need not to be fully omniscient. Precisely, let $\mathcal{S}_X = \{ i \in [n_X] | \phi_X(i) \ne \phi_Y(j) \textit{ and } P(i,j) > 0 \textit{ for some } j \}$ be the subset of $[n_X]$ for which may cause confusion. We then define the random variable $X_{\text{cut}}$ as the random variable $X$ restricted to that set $\mathcal{S}_X$, that is $X_\text{cut} = \mathbb{I}\{ X \in \mathcal{S}_X\} X$.
 Similarly, we can define the set $\mathcal{S}_Y$ and the random variable $Y_{\text{cut}}$.
 
\begin{theorem}\label{thm:limited_rate}
Let $\phi_X$ and $\phi_Y$ be functions on $[n_X]$ and $[n_Y]$ respectively. Then, the helper performs as good as an omniscient helper as long as:
\begin{align}
R_{X \to H} \geq H(X_\text{cut}), \quad R_{Y \to H} \geq H(Y_\text{cut}),
\end{align}
where $R_{X \to H}$ and $R_{Y \to H}$ are the rates from sources $X$ and $Y$ respectively to the helper.
\end{theorem}

\begin{proof}
Let the sources encode the sources $X_\text{cut}$ and $Y_\text{cut}$, and send them to the helper. Suppose $X \notin \mathcal{S}_X$, then by construction it must be that $\phi_X(X) = \phi_Y(Y)$, and therefore the helper can operate. If $X \in \mathcal{S}_X$, it must be that $Y \in \mathcal{S}_Y$ too, and therefore the helper can detect the errors.
\end{proof}

\noindent \textbf{Example:} 
Consider as an example the probability distribution represented below on the right. We let $\phi_X(i) = 1$ if $i = 1,2$, and $\phi_X(3) = -1$. On the other hand, we let $\phi_Y(j) = j$. The resulting rate region $\mathcal{R}_\phi$ derived from \eqref{eq:rates_non_symm} for this choice of functions is represented on the left, along with the Slepian-Wolf rate region $\mathcal{R}_{SW}$, and the G\'acs-K\"orner region $\mathcal{R}_{GK}$. In this case, the rate region $\mathcal{R}_\phi$ is larger than the Slepian-Wolf region, but it has to be emphasized that the helper has non-zero rate here! In fact, the rate of the helper is precisely $h(\frac{1}{2}\epsilon)$. Also note that the region $\mathcal{R}_\phi$ does not have a 45 degree slope dominant face.
\begin{minipage}{.49\textwidth}
\begin{minipage}[l]{.48\textwidth}
\centering
\includegraphics[scale=.35]{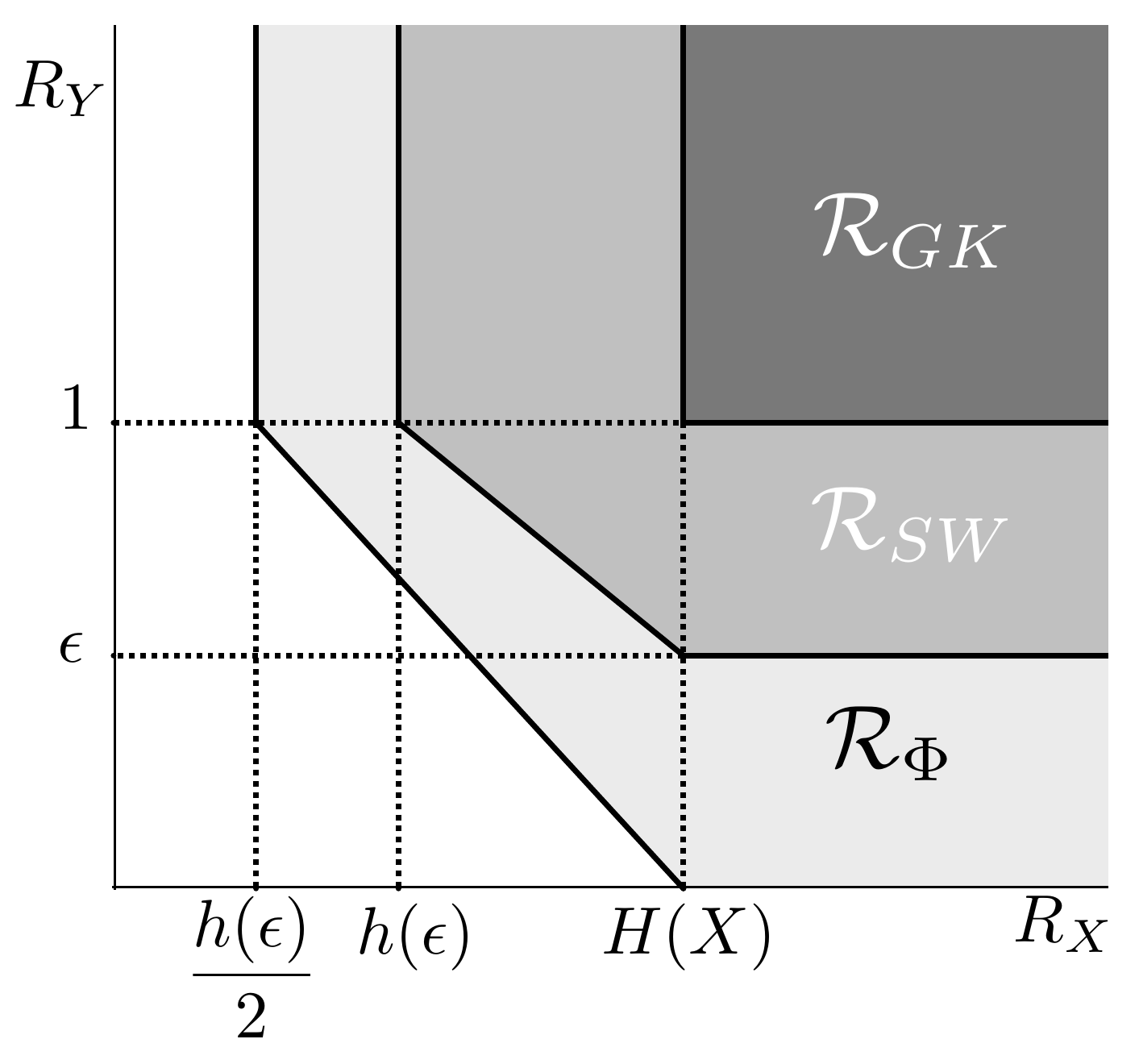}
\end{minipage} 
\begin{minipage}[c]{.45\textwidth}
\centering
\begin{tikzpicture}[scale=1, auto,swap]
		\foreach \pos/ \name/ \num in {{(0,0)/x_3/3}, {(0,1)/x_2/2}, {(0,2)/x_1/1}, {(1.5,0.66)/y_2/2},{(1.5,1.32)/y_1/1}}
        \node (\name) at \pos {$\num$};
        \node (cal_Y) at (1.7,1.8) {$\mathcal{Y}$};
        \node (cal_X) at (-0.2,2.4) {$\mathcal{X}$};
        \path[undirect_edge] (x_1) edge node[anchor = south, sloped] {\footnotesize $\frac{1}{2}(1 - \epsilon)$} (y_1);
        \path[undirect_edge] (x_3) edge node[anchor = north, sloped] {\footnotesize $\frac{1}{2}(1 - \epsilon)$} (y_2);
        \path[undirect_edge] (x_2) edge node[anchor = south, sloped] {\footnotesize $\frac{1}{2}\epsilon$} (y_1);
        \path[undirect_edge] (x_2) edge node[anchor = north, sloped] {\footnotesize $\frac{1}{2}\epsilon$} (y_2);
        \end{tikzpicture}
\end{minipage}
\end{minipage}

\begin{remark}
Denote by $E$ the output of the helper node. Then, $(I(Y;E|X),I(X;E|Y),I(X;Y|E)) \in \mathcal{I}(X;Y)$, where $\mathcal{I}(X;Y)$ is the region of tension defined in \cite{prabhakaran2014}. This region of tension is related to a notion of assisted common information, which is somewhat of a dual problem. Indeed, in the assisted common information problem, a helper node communicates with the sources directly, and helps them in agreeing on a common function. In our setting, the helper sends information to the terminal, and not to the sources.
\end{remark}

\subsection{Network Setting}

This type of approach to distributed coding generalizes nicely to the network setting, in particular correlated sources multicast. In this setting, the helper node is a front-end node who has large rate from the original sources, but limited rate towards the terminals. More precisely, consider a network $G$ with two sources $s_X$ and $s_Y$, and a terminal $t \in \mathcal{T}$. In addition, we consider a helper node $s_H$, which we assume is different than $s_X$ and $s_Y$ to avoid trivial cases.\\
To efficiently transmit X and Y to the terminals we follow these steps. First, we transmit X and Y to the front-end helper $s_H$. Then, we transmit $X$ and $Y$ to the terminal using source decomposition, with the help of $s_H$. Note that we exploit the fact that $\phi_X \approx \phi_Y$ to send the function once, and express $X$ and $Y$ with respect to the $\phi_X$ and $\phi_Y$ respectively. The resulting sufficient conditions are expressed below, where we only consider a corner point coding for notation purposes.

\begin{theorem}
Let $\phi_X$ and $\phi_Y$ be functions. Sources $X$ and $Y$ generated at nodes $s_X$ and $s_Y$ can be reliably and efficiently communicated through the network, if the following min-cut relations are satisfied:
\begin{align}
&\rho(s_X;s_H) \geq H(X), \quad  \rho(s_Y;s_H) \geq H(Y) \nonumber \\
&\rho(s_X,s_Y;s_H) \geq H(X)  + H(Y) \nonumber \\
&\rho(s_X;t) \geq H(X|\phi_X(X)) , \quad  \rho(s_Y;t) \geq H(Y) & \nonumber \\
&\rho(s_H;t) \geq H(\phi_X(X)|\phi_Y(Y)) \nonumber \\
&\rho(s_X,s_Y,s_H;t)  \geq H(X|\phi_X(X)) + H(Y|\phi_Y(Y)) + \nonumber \\ &\quad \qquad H(\phi_X(X),\phi_Y(Y))
\end{align}
\end{theorem}
\begin{proof}
The complete proof is omitted, but follows from the min-cut conditions for independent source coding. Note that the first three conditions characterize sufficient min-cut conditions from the source to the helper, while the last conditions characterize sufficient conditions from the sources and helper to the terminals. We refer the readers to \cite{technical} in which similar results are established.
\end{proof}

\begin{remark}
When the rate of the helper is 0, this reduces to the results of \cite{technical}, in which the common information $K$ between $X$ and $Y$ is sent only once to the terminal, and $X$ and $Y$ are described given that common information.
\end{remark}
\begin{remark}
Note that the rates to the helper $H(X)$ and $H(Y)$ can be reduced by using Theorem~\ref{thm:limited_rate}.
\end{remark}

\section{Optimal Decompositions}

So far, we have considered arbitrary decompositions of the the bipartite graph  of $p_{X,Y}$. In this section we aim at developing a notion of optimal decomposition. Note that there is an inherent trade-off in the construction of decomposition, as on one hand we wish to increase the entropy of the function $\phi_X(X)$, and on the other hand we want to make it \emph{agree} with $\phi_Y(Y)$ as often as possible. This trade-off is characterized by the reduction in rates from the sources, and the necessary rates from the helper, and is specified by the definition below:

\begin{definition}[Maximum Entropy Function]
We say a decomposition $X \rightarrow (X',\phi^*_X(X))$ and $Y \rightarrow (Y',\phi_Y^*(Y))$ is optimal for helper rate $\epsilon$ if it is the solution of the optimization:
\begin{align}
\text{maximize } & H(\phi_X(X)) \nonumber \\
\text{subject to } & H(\phi_X(X)|\phi_Y(Y)) \leq \epsilon \label{eq:optimal_dec}
\end{align}
where the optimization is taken over all functions $\phi_X$ and $\phi_Y$ taking values in a finite set.
\end{definition}

Note that the previous definition indeed captures the rates in Theorem~\ref{thm:arbitrary_func}, as maximizing $H(\phi_X(X))$ equivalently minimizes $H(X|\phi_X(X))$ by observing that $H(X) = H(\phi_X(X)) + H(X|\phi_X(X))$.

\begin{remark}
In general, the problem in \eqref{eq:optimal_dec} is not equivalent to the symmetric problem of maximizing $H(\phi_Y(Y))$ with constraint $H(\phi_Y(Y)|\phi_X(X)) \leq \epsilon$. However, for symmetric matrices $\bP$, \emph{i.e.} $\bP = \bP^T$, the two problems are equivalent. It then follows that the rate region under optimal decomposition \eqref{eq:rates_non_symm} will be symmetric for symmetric joint distributions. 
\end{remark}

An alternative formulation of the optimization in \eqref{eq:optimal_dec} is given by the Lagrangian relaxation:

\begin{definition}[Lagrangian Formulation]
The lagrangian formulation of \eqref{eq:optimal_dec} is, for some $\lambda > 0$:
\begin{align}
\text{maximize}_{\phi_X,\phi_Y} H(\phi_X(X)) - \lambda H(\phi_X(X) | \phi_Y(Y)) \label{eq:lagrangian}
\end{align}
\end{definition}

The formulation above is useful as it is an unconstrained optimization problem. Using this formulation, we first show that for large enough $\lambda$, the solution for this optimization problem is the G-K common information, and the helper is not required. We then focus on binary functions, and show a useful matrix representation as well. This allows us to harness spectral graph theory and suggest an approximation to the problem.

\begin{lemma}
There exist a $\lambda_{\text{max}}$ such that for any $\lambda>\lambda_{\text{max}}$, the solution of \eqref{eq:lagrangian} is the solution of \eqref{eq:gk_function_formulation}, \emph{i.e.} the optimal functions are the index of the disjoint components.
\end{lemma}

This Lagrangian relaxation lends itself to a simple matrix notation when $\phi_X$ and $\phi_Y$ take values in $\{-1,1 \}$. Recall that in that case, the necessary rate of the helper only encodes the error events, so it suffice to characterize the probability of error. We have:
\begin{align}
\PP (\phi_X(X) \neq \phi_Y(Y)) &= \sum_{i,j} p_{X,Y}(i,j) \frac{1}{2}(1 - \phi_X(i)\phi_Y(j)) \nonumber \\
& = \frac{1}{2}(1 - \phi_X^T P \phi_Y)
\end{align}
where $\phi_X$ and $\phi_Y$ are column vectors of length $n_X$ and $n_Y$ respectively, such that $(\phi_a)_{i} = \phi_a(i)$, for $a \in \{X,Y\}$ and $i \in [n_a]$. Similarly, the probability $\PP (\phi_X(X) = 1)$ can be written as $\frac{1}{2}(1 + \phi_X^T \bP \mathbf{1})$.

\begin{corollary}[Matrix notation] \label{cor:matrix_notation}
Let $\phi_X$ and $\phi_Y$ take values in $\{ -1,1\}$. Then, we can write \eqref{eq:lagrangian} as :
\begin{align}
\max_{\mathbf{\phi}_x,\mathbf{\phi_y}} h\left(\frac{1}{2}(1 + \bm{\phi_X}^T \bP \mathbf{1}) \right) - \lambda h\left(\frac{1}{2}(1 - \bm{\phi_X}^T \bP \bm{\phi_Y})\right) \label{eq:matrix_not}
\end{align}
where the maximization is over vectors $\phi_X$ and $\phi_Y$ of size $n_x$ and $n_y$ respectively, where each coordinate is in $\{ -1,1\}$.
\end{corollary}

This matrix notation allows for some interesting insights. First, suppose there is a solution that has probability of error $P_e \vcentcolon= \frac{1}{2}( 1 - \phi_X^T \bP \phi_Y)$. Then, by considering $\tilde{\phi}_Y = - \phi_Y$, we obtain a solution that has probability of error $1 - P_e$. As the binary entropy function is symmetric, these two probability of errors lead to an identical rate from the helper, and therefore, there is no loss in only considering solutions such that the probability of error is less than $\frac{1}{2}$. On the other hand, observe that $h\left( \frac{1}{2}(1 + \phi_X^T \bP \mathbf{1}) \right)$ takes its maximum when $\phi_X^T \bP \mathbf{1}$ is close to 0. Therefore, there is once again no loss in considering $\phi_X$ such that $\PP(\phi_X(X)=1) \leq \frac{1}{2}$, as $\tilde{\phi}_X(X) = -\phi_X(X)$ achieves $\PP(\tilde{\phi}_X(X) = 1) = 1 - \PP(\phi_X(X)=1)$. These two observations along with the monotonicity of $h(\cdot)$ allow to simplify the search space, and lead to the following result:

\begin{corollary}\label{cor:conductance}
There exist a $\lambda>0$ such that the solution of \eqref{eq:matrix_not} is also the solution of:
\begin{align*}
\text{max}_{\phi_X,\phi_Y}\frac{\PP(\phi_X(X) = 1)}{\PP(\phi_X(X) \neq \phi_Y(Y))}, \,
\text{s.t.} \quad \PP(\phi_X(X)=1) \leq \frac{1}{2}
\end{align*}
\end{corollary}

\begin{proof}
The complete proof is omitted but follows essentially from a linearization, and a Dinkelbach decomposition \cite{dinkelbach1967}, which allows to transform the optimization problem \eqref{cor:matrix_notation} into an equivalent fractional optimization (for some $\lambda$).
\end{proof}

\noindent The optimization in Corollary~\ref{cor:conductance} is similar to a graph conductance formulation. Indeed, the term $\PP(\phi_X(X) \neq \phi_Y(Y))$ can be interpreted as the value of a cut on the bipartite graph for the partition $\mathcal{S} = \{(X,Y)| \phi_X = \phi_Y = 1\}$, while the value $\PP(\phi_X(X) = 1)$ is related to the size of the partition $\mathcal{S}$. This relationship can be made precise in the case of symmetric matrices $\bP$, and is the object of some future work.

\subsection{Spectral algorithm for finding optimal decomposition}

In this section, we present an approximation algorithm to optimal decomposition. Indeed, in general the problem is combinatorial and unless $\lambda$ is large enough, in which case it reduces to finding disjoint components in the bipartite graph, it is a hard problem. In particular, Corollary~\ref{cor:conductance} suggests that this problem is related to the graph conductance problem in some contexts, which is a hard problem to solve exactly. However, a reasonable approach is to look at a spectral clustering approach, in which the goal is to find a well-balanced (in terms of size) partition of nodes such that they are well separated (the cut between the partition is small). This is motivated by results in spectral graph theory, namely Cheeger's Inequality \cite{spielman2007spectral}, which relate the value of the conductance of the graph to the second eigenvalue of the so-called normalize Laplacian matrix. However, because of the special structure of our graph, namely a bipartite graph describing a joint distribution, we will exhibit a relationship between the normalized Laplacian matrix, and the spectral decomposition in Theorem~\ref{thm:max_corr_def}, and show that the second eigenvalue of the normalized laplacian is in fact the maximal correlation coefficient.

Before we proceed, let us recall the definition of the normalized Laplacian matrix. Let $G = (\mathcal{V},\mathcal{E}, w)$ be an undirected graph, with $|\mathcal{V}| = n$ nodes, and weights on the edges $w_{i,j}$. We let $\bA$ be the $n \times n$ adjacency matrix defined as $(\bA)_{i,j} = w_{i,j}$. Let $\bD$ be the diagonal matrix with diagonal entries $(\bD)_{i,i} = \sum_{i,j} w_{i,j}$. Then, the normalized Laplacian matrix $\bN$ is defined as $\bN \vcentcolon = \bI - \bD^{-1/2} \bA \bD^{-1/2}$.
The following theorem links the second smallest eigenvalue of $\bN$ for the bipartite representation of $p_{X,Y}$, with $\rho^2_m(X;Y)$ .
\begin{theorem}
Denote by $\nu$ the second smallest eigenvalue of $\bN$. Then, $\nu = 1 - \rho_m^2(X;Y)$.
\end{theorem}

\begin{proof}
Noticing that for the bipartite graph, the adjacency matrix $\bA$ and the corresponding $\bD$ can be written as:
\begin{align*}
& \bA  = \left[ \begin{array}{cc}
0 & \bP \\
\bP^T & 0
\end{array} \right] , \quad 
\bD  = \left[ \begin{array}{cc}
\bD_X & 0 \\
0 & \bD_Y
\end{array} \right]. \nonumber \\
\Rightarrow & \bD^{-1/2} \bA \bD^{-1/2} = \left[ \begin{array}{cc}
0 & \bD_X^{-1/2} \bP \bD_Y^{-1/2} \\
\bD_Y^{-1/2} \bP^T \bD_X^{-1/2} & 0
\end{array} \right] 
\end{align*}
Therefore, the eigenvalues of $\bD^{-1/2} \bA \bD^{-1/2}$ are the plus or minus singular values of $\bQ = \bD_X^{-1/2} \bP \bD_Y^{-1/2}$, and the second largest singular value of $\bQ$ must be equal to $1 - \nu$.
\end{proof}
This theorem relates the maximal correlation coefficient to $\nu$, and therefore indirectly to the conductance of the bipartite graph. In fact, the following corralary is immediate, and shows that when there exist disjoint components, the maximal correlation is precisely 1.

\begin{corollary}\label{thm:max_corr_gk}
Let $\rho_m(X;Y)$ be the maximal correlation of $p_{X,Y}$, and denote by $K$ the G-K common information. Then, $H(K) \neq 0$ if and only if $\rho_m(X;Y)  = 1$. In addition, call $m$ the multiplicity of the singular value $1$ in the singular value decomposition of $\bQ$. Then $H(K) \leq \log (m)$.
\end{corollary}

In view of these results, we propose a simple and efficient algorithm for finding binary functions $\phi_X$ and $\phi_Y$. We compute the left and right singular vectors $u$ and $v$, respectively, corresponding to the second highest singular value of $\bQ$. Then, we chose a threshold and let $\phi_X(i) = 1$ if $u_i > t$ and $\phi_X(i) = -1$ otherwise, and similarly for $\phi_Y(j)$ and $v_j$. Note that for an appropriate choice of $t$, this algorithm finds disjoint components in the bipartite graph, if they exist.
\begin{remark}
The above algorithm can be iteratively applied to each partition in order to find functions $\phi_X$ and $\phi_Y$ that have non-binary range, or equivalently to find more \emph{clusters}.
\end{remark}
\section{Concluding Remarks}

The optimization problem in \eqref{eq:lagrangian} leads to a couple of interesting observations. First, we have only considered the case where the function $\phi_X$ is a function of one single symbol $X$. A natural extension would be to consider the product distribution $p_{X^n,Y^n}$ and functions that take as input $n$ symbols. In contrast with the case of G\'acs-K\"orner common information, in which taking the product distribution could not help as the number of disjoint components would not increase, the optimization in \eqref{eq:lagrangian} will likely improve by taking product distributions as it increases the options available for edge removal. Another interesting open question is related to the relationship between the maximal correlation and the optimal source decomposition. In particular, we have shown that the maximal correlation coefficient can be used as an approximation to optimal decomposition, but only for some specific $\lambda$. A question of interest would be to determine whether a similar spectral method can be used to approximate the trade-off fully.

\bibliographystyle{./biblio/IEEEtran}
\bibliography{./biblio/IEEEabrv,references}

\end{document}